\documentclass[conference,10pt]{IEEEtran}

\usepackage{amsmath,amssymb,amsthm}
\usepackage{xifthen}
\usepackage{mathtools}
\usepackage{enumerate}
\usepackage{microtype}
\usepackage{xspace}
\usepackage{bm}
\usepackage[T1]{fontenc}
\usepackage{fancyhdr}
\usepackage{lastpage}
\usepackage{bbm}

\newcommand{\mbs}[1]{\pmb{#1}}
\newcommand{\vect}[1]{{\lowercase{\mbs{#1}}}}
\newcommand{\mat}[1]{{\uppercase{\mbs{#1}}}}

\newcommand{\T}{{\scriptscriptstyle\mathsf{T}}}
\renewcommand{\H}{{\scriptscriptstyle\mathsf{H}}}

\renewcommand{\Re}[1][]{\ifthenelse{\isempty{#1}}{\operatorname{Re}}{\operatorname{Re}\left(#1\right)}}
\renewcommand{\Im}[1][]{\ifthenelse{\isempty{#1}}{\operatorname{Im}}{\operatorname{Im}\left(#1\right)}}

\newcommand{\SNR}{\mathsf{snr}}

\newcommand{\rv}{\vect{r}}

\newcommand{\uv}{\vect{u}}
\newcommand{\vv}{\vect{v}}

\newcommand{\xv}{\vect{x}}
\newcommand{\yv}{\vect{y}}

\newcommand{\Sigmam}{\pmb{\Sigma}}

\newcommand{\Am}{\mat{a}}

\newcommand{\Mm}{\mat{M}}

\newcommand{\Um}{\mat{u}}
\newcommand{\Vm}{\mat{V}}

\newcommand{\Cc}{{\mathcal C}}
\newcommand{\Dc}{{\mathcal D}}

\newcommand{\Nc}{{\mathcal N}}

\newcommand{\Pc}{{\mathcal P}}
\newcommand{\Qc}{{\mathcal Q}}
\newcommand{\Rc}{{\mathcal R}}

\newcommand{\CC}{\mathbb{C}}

\newcommand{\Id}{\mat{\mathrm{I}}}

\newcommand{\CN}[1][]{\ifthenelse{\isempty{#1}}{\mathcal{N}_{\mathbb{C}}}{\mathcal{N}_{\mathbb{C}}\left(#1\right)}}

\renewcommand{\P}[1][]{\ifthenelse{\isempty{#1}}{\mathbb{P}}{\mathbb{P}\left(#1\right)}}
\newcommand{\E}[1][]{\ifthenelse{\isempty{#1}}{\mathbb{E}}{\mathbb{E}\left[#1\right]}}
\newcommand{\I}[1][]{\ifthenelse{\isempty{#1}}{\mathbb{I}}{\mathbb{I}\left\{#1\right\}}}
\renewcommand{\det}[1][]{\ifthenelse{\isempty{#1}}{\mathrm{det}}{\text{det}\left(#1\right)}}
\newcommand{\trace}[1][]{\ifthenelse{\isempty{#1}}{\mathrm{tr}}{\text{tr}\left(#1\right)}}
\newcommand{\rank}[1][]{\ifthenelse{\isempty{#1}}{\mathrm{rank}}{\text{rank}\left(#1\right)}}
\newcommand{\diag}[1][]{\ifthenelse{\isempty{#1}}{\mathrm{diag}}{\text{diag}\left(#1\right)}}
\newcommand{\Cov}[1][]{\ifthenelse{\isempty{#1}}{\mathsf{Cov}}{\mathsf{Cov}\left(#1\right)}}

\newcommand{\defeq}{\triangleq}

\newtheorem{proposition}{Proposition}
\newtheorem{remark}{Remark}[section]
\newtheorem{theorem}{Theorem}
\newtheorem{lemma}{Lemma}

\newcounter{enumi_saved}
\usepackage{answers}
\Newassociation{solution}{Solution}{solutionfile}

\AtBeginDocument{\Opensolutionfile{solutionfile}[\jobname]}
\AtEndDocument{\Closesolutionfile{solutionfile}\clearpage
}

\IfFileExists{MinionPro.sty}{
}{
}

\usepackage{subfigure}
\usepackage{multirow}
\usepackage{array}
\usepackage{amsmath}
\usepackage{color}
\usepackage{enumitem}
\usepackage{stfloats}
\usepackage{tikz,pgfplots}
\usetikzlibrary{shapes}
\usetikzlibrary{plotmarks}
\usetikzlibrary{spy}

\renewcommand{\rv}[1]{{\mathrm{#1}}}
\newcommand{\rvVec}[1]{{\pmb{\mathrm{#1}}}}
\newcommand{\rvMat}[1]{{\pmb{\mathsf{#1}}}}
\newcommand{\const}{c_0}

\renewcommand{\SNR}{\mathsf{SNR}}

\renewcommand{\defeq}{:=}

\newcommand{\Yt}{{\tilde{\rvMat{Y}}}}
\newcommand{\Xvt}{{\tilde{\rvVec{X}}}}
\newcommand{\Xt}{{\tilde{\rv{X}}}}
\newcommand{\Zt}{{\tilde{\rvMat{Z}}}}

\title{The Optimal DoF Region for the Two-User Non-Coherent SIMO Multiple-Access Channel} 

\author{\IEEEauthorblockN{Khac-Hoang Ngo\IEEEauthorrefmark{1}\IEEEauthorrefmark{2}, Sheng Yang\IEEEauthorrefmark{1}, Maxime Guillaud\IEEEauthorrefmark{2}}
	\IEEEauthorblockA{\IEEEauthorrefmark{1}LSS, CentraleSup\'elec, 91190 Gif-sur-Yvette, France \\
		\IEEEauthorrefmark{2}Mathematical and Algorithmic Sciences Lab, Paris Research Center, Huawei Technologies, \\92100 Boulogne-Billancourt, France \\
		Email: {\tt \{ngo.khac.hoang, maxime.guillaud\}@huawei.com, sheng.yang@centralesupelec.fr}}
}

\begin{document}
	
\maketitle
\date{\today}
\begin{abstract}
  The optimal degree-of-freedom~(DoF) region of the non-coherent
  multiple-access channels is still unknown in general. In this paper, we make
  some progress by deriving the \emph{entire} optimal DoF region in the case of
  the two-user single-input multiple-output~(SIMO) generic block
  fading channels. The achievability is based on a simple training-based
  scheme. The novelty of our result lies in the converse using a 
  genie-aided bound and the duality upper bound. As a by-product, our
  result generalizes previous proofs for the single-user Rayleigh block
  fading channels. 
\end{abstract}


\section{Introduction}

The fundamental limit of communication over wireless fading channels
depends on the availability of channel state information~(CSI) at the
transmitter/receiver. While the channel statistics are normally stable
and can be assumed to be available, the assumption on {\em
instantaneous} CSI varies with the context. When the instantaneous CSI
is assumed to be {\em a priori} known, e.g., in fixed environments where
it changes slowly and can be estimated accurately at negligible cost, at
least at the receiver side, the communication is said to be {\em coherent}. On the other hand, if the instantaneous CSI is {\em a priori} unknown, e.g, when the estimation cost is not negligible, the communication is said to be {\em non-coherent}.  

In a point-to-point multiple-input multiple-output~(MIMO) channel with
$M$ transmit and $N$ receive antennas, it is well known that the {\em
coherent} capacity scales linearly with the number of antennas as $C\sim
\min\left\{ M, N \right\} \log \SNR$ at high signal-to-noise
ratio~(SNR)~\cite{Foschini,Telatar1999capacityMIMO}. 
In the {\em non-coherent} case with stationary fading, the capacity scales as
$\log\log \SNR + \chi(\rvMat{H}) + o(1)$\footnote{$\chi(\rvMat{H})$ is
called the fading number of the channel.}~\cite{Moser}, implying a DoF
of $0$. Nevertheless, if the channel remains constant during a certain
amount of slots, say $T$ slots, then the DoF becomes strictly positive
as $M^*(1-\frac{M^*}{T})$ where $M^* \defeq \min\left\{ M, N, \lfloor
\frac{T}{2} \rfloor \right\}$. This fading setup is commonly referred to 
as the block fading channel, and has been extensively investigated in
the literature~\cite{Hochwald2000unitaryspacetime, ZhengTse2002Grassman,
Yang2013CapacityLargeMIMO}. 
Remarkably, in the block fading case, the optimal DoF can be achieved
either by well-designed space-time
modulations~\cite{Hochwald2000unitaryspacetime,ZhengTse2002Grassman,
Yang2013CapacityLargeMIMO}, or by simple training-based
strategies~\cite{Hassibi2003howmuchtraining}. The converse in the
aforementioned works was based on the Rayleigh fading assumption, using
either a direct approximation at high SNR~\cite{ZhengTse2002Grassman} or
a duality upper bound with a carefully chosen auxiliary output
distribution~\cite{Yang2013CapacityLargeMIMO}. 

In multi-user MIMO channels, such as the broadcast channels~(BC) and the multiple access channels~(MAC), non-coherent communications have been studied in the block fading case.~
For the BC, the exact DoF region is known with \emph{isotropic} Rayleigh fading (a special case of stochastically degraded BC) and can be achieved with time division multiple access~(TDMA)~\cite{Fadel2016coherencedisparity}. Some achievable schemes have been proposed for the BC with spatially correlated fading~\cite{Hoang2017BC,Zhang2017spatiallyCorrelatedBC}. For the MAC, it has been shown that the optimal sum DoF can be achieved with a training-based scheme~\cite{Fadel2016coherencedisparity}, but the optimal DoF region is still {\em unknown}.

In this work, we make some progress for the non-coherent single-input multiple-output~(SIMO) MAC. Specifically, we derive the optimal DoF region in the case of two single-antenna transmitters (users) and a $N$-antenna receiver in block fading channel with coherence time $T$. When $N=1$, the region is achieved with a simple time division multiplexing between two users. In this case, letting two users cooperate does not help exploit more degrees of freedom and it is optimal to activate only one user at a time to achieve $1-\frac{1}{T}$ DoF for that user. When $N>1$, a training-based scheme can achieve another DoF pair. We let two users send orthogonal pilots for channel estimation in the first $2$ time slot, then send data simultaneously in the remaining $T-2$ time slots. In this way, each user can achieve $1-\frac{2}{T}$ DoF. 

The main technical contribution of this paper lies in the converse
proof. Leveraging the duality upper bound~\cite{Moser}, we carefully
choose an output distribution with which we derive a tight outer bound
on the DoF region. Unlike previous results such as \cite{ZhengTse2002Grassman,
Yang2013CapacityLargeMIMO}, we do not assume the Gaussianity of the
channel coefficients, which makes our proof more general and our results
stronger even in the single-user case.

The remainder of this paper is organized as follows. The system model and preliminaries are presented in Section~\ref{sec:model}. In Section~\ref{sec:result}, we provide the main
result on the optimal DoF region of the two-user MAC, as well as the proof
for the case $N=1$ and the achievability for the case $N>1$. We
introduce the converse proof technique through a new proof for the single-user SIMO channel in Section~\ref{sec:singleUser}, and use it to show the tight outer bound for the case $N>1$ of the MAC in Section~\ref{sec:converse}. Finally, we conclude the paper in Section~\ref{sec:conclusion}. 

{\it Notations:} For random quantities, we use
upper case non-italic letters: normal fonts, e.g., $\rv{X}$, for scalars; bold fonts,
e.g., $\rvVec{V}$, for vectors; and bold and sans serif fonts, e.g.,
$\rvMat{M}$, for matrices. Deterministic quantities are denoted 
with italic letters, e.g., a scalar $x$, a vector $\pmb{v}$, and a
matrix $\pmb{M}$. Throughout the paper, we adopt the column convention
for vectors. The Euclidean norm of a vector and a matrix is denoted by $\|\vv\|$ and $\|\Mm\|$, respectively. 
The 
transpose and conjugated transpose of $\Mm$ is 
$\Mm^\T$ and $\Mm^\H$, respectively. $\rvMat{M}_{[i:j]}$ denotes the sub-matrix containing columns from $i$ to $j$ of a matrix $\rvMat{M}$ (thus $\rvMat{M}_{[i]}$ denotes column $i$). $\diag[x_1,\dots,x_N]$ denotes the diagonal matrix with diagonal entries $x_1,\dots,x_N$. 
$H(.)$, $h(.)$, and $D(.\|.)$ denote the entropy, differential entropy, and Kullback-Leibler divergence, respectively. Logarithms are in base $2$. 
$(x)^+ = \max\{x,0\}$. ``$\defeq$'' means ``is defined as''. $\Gamma(x) = \int_{0}^{\infty}z^{x-1}e^{-z}dz$ is the Gamma function. 
Given two functions $f$ and $g$, we write $f(x) = O(g(x))$ if there exists a constant $c>0$ and some $x_0$ such that $f(x) \le cg(x), \forall x\ge x_0$.

\section{System Model and Preliminaries} \label{sec:model}
We consider a single-input multiple-output~(SIMO) multiple-access
channel in which two single-antenna users send their signals to a
receiver with $N$ antennas. The channel between the users and the
receiver is flat and block fading with equal and synchronous coherence
interval of $T$ symbol periods. That is, the channel vector $\rvVec{H}_k
\in \CC^{N\times1}$, $k=1,2$, remains unchanged during each block of length $T$ symbols and changes independently between blocks. The realizations of $\rvVec{H}_1$ and $\rvVec{H}_2$ are {\em unknown} to both the users and the receiver. 
The received signal during the coherence block~$b$, $b=1,2,\ldots$\footnote{Throughout,
we omit the block index whenever confusion is unlikely.}, is 
\begin{align}
\rvMat{Y}[b] = \rvVec{H}_1[b]\, \rvVec{X}_1^\T[b] + \rvVec{H}_2^\T[b] \, \rvVec{X}_2^\T[b]+ \rvMat{Z}[b],
\end{align}
where $\rvVec{X}_1 \in \CC^{T}$ and $\rvVec{X}_2 \in \CC^{T}$ are the transmitted signals from user 1 and user 2, respectively, with the power constraint 
\begin{align} \label{eq:powerConstraints}
  \frac{1}{B} \sum_{b=1}^B \|\rvVec{X}_{i}[b]\|^2 \le P T, \quad i = 1,2,
\end{align}
where $B$ is the number of the blocks spanned by a codeword. 
We assume that $\rvMat{Z} \in \CC^{N\times T}$ is the additive white
Gaussian noise with independent and identically distributed~(i.i.d.)~$\Cc\Nc(0,1)$ entries. 
The parameter $P$ is the average power ratio between the transmitted signal and the noise, thus we refer to $P$ as the SNR of the channel. 

Since the channel is block memoryless\footnote{The results can be
generalized to stationary fading as done in \cite{Moser}.}, it is well known that a rate pair $(R_1(P),R_2(P))$ in bits per channel
use is achievable at SNR $P$, i.e., lies within the capacity region $\mathcal{C}_{\text{Avg}}(P)$, for the MAC if and only if
\begin{align}
R_1+R_2 &\le \frac{1}{T}I(\rvVec{X}_1, \rvVec{X}_2; \rvMat{Y}), \\
R_1 &\le \frac{1}{T}I(\rvVec{X}_1; \rvMat{Y} | \rvVec{X}_2), \label{eq:boundR1}\\
R_2 &\le \frac{1}{T}I(\rvVec{X}_2; \rvMat{Y} | \rvVec{X}_1), \label{eq:boundR2}
\end{align}
for some input distribution subject to the average power constraint $P$ (as the codeword length $B$ goes to infinity)~\cite{ElGamal}. 
Then, we say that $(d_1,d_2)$ is an achievable DoF pair with
\begin{align}
d_k \defeq \liminf\limits_{P\to\infty} \frac{R_k(P)}{\log(P)}, \quad k = 1,2.
\end{align} 
The optimal DoF region $\Dc_{\text{Avg}}(P)$ is defined as the set of all achievable DoF pairs. 

We assume that the channel vectors $\rvVec{H}_1$ and $\rvVec{H}_2$ are
independent\footnote{Independence is not necessary but makes the
analysis slightly simpler.} and drawn from a generic distribution satisfying the following conditions:
\begin{align}
  h(\rvVec{H}_k) &> -\infty, \quad \E[\| \rvVec{H}_k \|^2] < \infty,
  \quad k=1,2. 
\end{align}%
The following results, whose proofs are provided in Appendix~\ref{app:proofs}, are useful for our main analysis. 
\begin{lemma}\label{lemma:1}
  Let $\Am\in\mathbb{C}^{m\times t}$ have full column rank,
  $\rvMat{W}\in\mathbb{C}^{n\times m}$ be such that
  $h(\rvMat{W})>-\infty$ and $\E[\|\rvMat{W}\|_{\rm F}^2] < \infty$, then we have
  \begin{align}
    h(\rvMat{W}\Am) &= n \log\det(\Am^\H\Am) + \const
  \end{align}%
  where $\const$ is bounded by some constant that only depends on the
  statistics of $\rvMat{W}$.  
\end{lemma}

\begin{lemma}\label{lemma:2}
  Let $\rv{X}\ge0$ be some random variable such that $\E[\rv{X}]<\infty$
  and $h(\rv{X}/\E[\rv{X}]) > -\infty$. Then, for any $\alpha<1$, 
  \begin{align}
    \E[\log(1+\rv{X})] &\ge \alpha \log(1+\E[\rv{X}]) + \const \label{eq:lemma2}
  \end{align}%
  where $\const>-\infty$ is some constant that only depends on $\alpha$. 
\end{lemma}
From the above result, we observe that 
when $\E[\rv{X}]\to\infty$,
$\frac{\E[\log(1+\rv{X})]}{\log(1+\E[\rv{X}])} \approx 1$ since we can
let $\alpha$ be arbitrarily close to $1$. The upper bound is simply from
Jensen's inequality. 

If the support of the input distribution is further bounded such that $\|\rvVec{X}_i\|^2 \le P$, $i=1,2$, then we say that the input satisfies the peak power constraint $P$. In this case, the capacity region and DoF region are denoted $\Cc_{\text{Peak}}(P)$ and $\Dc_{\text{Peak}}(P)$, respectively. Since the peak power constraint implies the average power constraint, we have that
\begin{align}
\Cc_{\text{Peak}}(P) \subseteq \Cc_{\text{Avg}}(P), \quad \Dc_{\text{Peak}}(P) \subseteq \Dc_{\text{Avg}}(P).
\end{align}

\begin{lemma}\label{lemma:3}
  For any rate pair $(R_1,R_2)$ achievable under the
  average power constraint $P$, for any $\beta\!>\!1$, there exists $(R_1',\!R_2')$ achievable
  under the peak power constraint $P^\beta$, such that 
  \begin{align}
    R_k - R_k' = O(P^{1-\beta} \log P^{\beta}), \quad k=1,2, \label{eq:boundRpeak}
  \end{align}%
  In short, 
  \begin{align}
    \mathcal{C}_{\text{Avg}}(P) &\subseteq \mathcal{C}_{\text{Peak}}(P^\beta) + O(P^{1-\beta} \log P^{\beta}), \quad \forall\,\beta>1. \label{eq:boundCpeak}
  \end{align}%
\end{lemma}
Since the pre-log of the gap $P^{1-\beta} \log P^{\beta}$ is vanishing at high SNR for any $\beta>1$, we have the DoF region
\begin{align}
  \mathcal{D}_{\text{Avg}}(P) &\subseteq \mathcal{D}_{\text{Peak}}(P^\beta) 
  \subseteq \mathcal{D}_{\text{Avg}}(P^\beta)
  , \quad \forall\,\beta>1. 
\end{align}%
Letting $\beta$ arbitrarily close to $1$, we conclude that using the peak power
constraint instead of the average power constraint does not change the
optimal DoF region. We therefore consider throughout the peak power
constraint, which can simplify considerably the analysis. 

\begin{lemma} \label{lemma:distributionR}
  Let $\rvVec{Y} \in \CC^{N}$ be a vector-valued random variable with distribution $\Pc$. Consider another family of distributions $\Rc$ whose densities are given by
	\begin{align}
	r_\rvVec{Y}(\yv) = \frac{\Gamma(N) |\det \pmb{A}|^2}{\pi^{N} \beta^\alpha \Gamma(\alpha)} \|\pmb{A} \yv\|^{2(\alpha-N)} \exp\left(-\frac{\|\pmb{A} \yv\|^2}{\beta}\right), \label{eq:distR}
	\end{align}
	for $\yv \in \CC^{N}$, where $\alpha, \beta > 0$, $\pmb{A}$ is any nonsingular deterministic $N\times N$ complex matrix. When $\beta = \E_\Pc[\|\pmb{A} \rvVec{Y}\|^2]$ and $\alpha = 1/\log(\beta) = 1/\log(\E_\Pc[\|\pmb{A} \rvVec{Y}\|^2])$, denote this distribution as $\Rc(N,\pmb{A})$. In this case,
	\begin{multline}
	\E_\Pc[-\log(r_\rvVec{Y}(\rvVec{Y}))] = -\log |\det\pmb{A}|^2 + N \E_\Pc[\log \|\pmb{A} \rvVec{Y}\|^2] \\+ O(\log\log(\E[\|\Am\rvVec{Y}\|^2])). \label{eq:boundDistR}
	\end{multline}
\end{lemma}

If we take $\rvVec{Y}$ as the channel output, as long as $\E[\|\Am\rvVec{Y}\|^2] \le P^{c_0}$ for any constant $c_0$ whose value only depends on the channel statistics, the term $O(\log\log(\E[\|\Am\rvVec{Y}\|^2]))$ scales double-logarithmically with $P$. 
Therefore, in the DoF sense, it is enough to consider only the first two terms in~\eqref{eq:boundDistR}.

\section{Main Result} \label{sec:result}
The main finding of this paper is the optimal DoF region of the MAC described above, as stated in Theorem~\ref{theo:DoFregion}.
\begin{theorem} \label{theo:DoFregion}
	For the non-coherent multiple-access channel with two single-antenna transmitters and a $N$-antenna receiver in flat and block fading with coherence time $T$, the optimal DoF region is characterized by
	\begin{align}
	d_1 + d_2 \le 1-\frac{1}{T},
	\end{align}
	if $T \le 2$ or $N=1$, and
	\begin{align}
	\frac{d_1}{T-2}+d_2 &\le 1-\frac{1}{T}, \label{eq:DoFbound1}\\
	d_1+\frac{d_2}{T-2} &\le 1-\frac{1}{T}, \label{eq:DoFbound2}
	\end{align}
	otherwise. 
\end{theorem} 
\begin{remark}
	When $T\to \infty$, the optimal DoF region approaches the region in the coherent case: $d_1+d_2 \le 1$ if $N=1$, and $\max\{d_1,d_2\}\le 1$ if $N>1$ (as shown in~Figure~\ref{fig:DoF_SIMO_MAC_2users}). 
\end{remark}	
\begin{figure}[!h] 
	\centering
	\includegraphics[width=.4\textwidth]{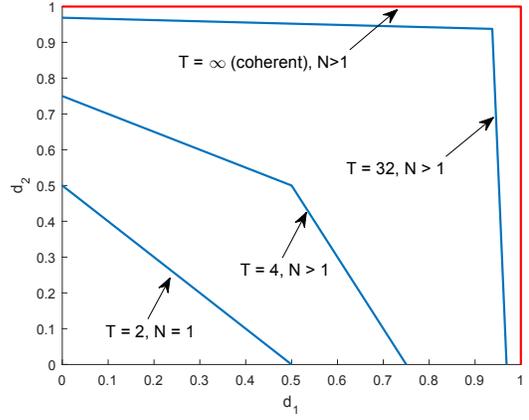}
	\caption{The optimal DoF region of two-user SIMO MAC with $N$ receive antennas in block fading with coherence time $T$.}
	\label{fig:DoF_SIMO_MAC_2users}
\end{figure}

The case $T = 1$~(stationary fading) is trivial: zero DoF is achievable, even if two users cooperate~\cite{Moser}. If $T = 2$ or $N=1$, the optimal DoF region is achieved with time division multiplexing between the users, noting that the active user can achieve $1-\frac{1}{T}$ DoF by either a training-based scheme~\cite{Hassibi2003howmuchtraining} or unitary space-time modulations~\cite{Hochwald2000unitaryspacetime,ZhengTse2002Grassman}. The tight outer bound follows by letting two users cooperate, then according to~\cite{ZhengTse2002Grassman,Yang2013CapacityLargeMIMO}, it is optimal to use $\min\{2,N,\lfloor \frac{T}{2}\rfloor\} = 1$ transmit antenna and achieve $1-\frac{1}{T}$ DoF in total. 


When $T\ge 3, N>1$, the region is the convex hull of the origin and three points: $\left(1-\frac{1}{T}, 0\right)$, $\left(0,1-\frac{1}{T}\right)$, and $\left(1-\frac{2}{T},1-\frac{2}{T}\right)$. The first two points are achieved by activating only one user. The third point is achieved with a training-based scheme: let two users send orthogonal pilots in the first two time slot for the receiver to learn their channel, then send data in the remaining $T-2$ time slots. The region is then achieved with time sharing between these points. It remains to show the tight outer bound for this case $T\ge 3, N>1$, but before that, let us introduce the proof technique by using it for a new proof of the tight DoF for the single-user SIMO channel in the next section.

\section{Single-User SIMO Channel Revisited} \label{sec:singleUser}
Consider the single-user (point-to-point) SIMO channel with block fading
with coherence time $T$ 
\begin{align}
\rvMat{Y} = \rvVec{H}\,\rvVec{X}^\T + \rvMat{Z},
\end{align}
where we have the same assumptions as in the MAC channel. 
It was shown that the DoF of this channel is $1-\frac{1}{T}$ and can be
achieved with either a training-based
scheme~\cite{Hassibi2003howmuchtraining} or well-designed space time
modulations~\cite{Hochwald2000unitaryspacetime,ZhengTse2002Grassman,Yang2013CapacityLargeMIMO}.
For the converse of the high SNR capacity (which implies the converse of the
DoF), while $h(\rvMat{Y} | \rvVec{X})$ can be calculated easily, the
upper bound for $h(\rvMat{Y})$ is much more involved
~\cite{ZhengTse2002Grassman,Yang2013CapacityLargeMIMO}. In this section,
we provide a simpler proof for the converse of the DoF using the duality approach as in~\cite{Yang2013CapacityLargeMIMO} but with a simple choice of auxiliary output distribution.

First, let us define the random variable $\rv{V}$ as the index of the
strongest input component, i.e.,\footnote{When there are more than one such components, we pick an arbitrary one.}
\begin{align}
\rv{V} \defeq \arg\max_{i=1,2,\dots,T}
|\rv{X}_i|^2.
\end{align} 
Thus, $\rv{X}_\rv{V}$ denotes the entry in $\rvVec{X}$ with the largest magnitude. Let the genie give $\rv{V}$ to the
receiver,\footnote{This technique of giving the index of the strongest input
component to the receiver was initially proposed in~\cite{ShengIT2017phasenoise} for phase noise channel.} we have
\begin{align}
I(\rvVec{X};\rvMat{Y}) &\le I(\rvVec{X};\rvMat{Y}, \rv{V})\\
&= I(\rvVec{X};\rvMat{Y} | \rv{V}) + I(\rvVec{X}; \rv{V} ) \\
&\le h(\rvMat{Y} | \rv{V}) - h(\rvMat{Y} | \rvVec{X}, \rv{V}) + H(\rv{V}) \\
&\le h(\rvMat{Y} | \rv{V}) - h(\rvMat{Y} | \rvVec{X}) + \log(T), \label{eq:boundRp2p}
\end{align}
where the last inequality is because we have the Markov chain $\rv{V}
\leftrightarrow \rvVec{X} \leftrightarrow \rvMat{Y}$ and $H(\rv{V}) \le
\log(T)$. For a given $\rvVec{X}$, we can apply Lemma~\ref{lemma:1} with
$\rvMat{W}=[\rvMat{H}\ \rvMat{Z}]$ and $\Am=\left[ 
  \rvVec{X} \ \Id_T \right]^\T$ to 
obtain
\begin{align}
h(\rvMat{Y} | \rvVec{X}) &= N\E[\log\det(\Id_T + \rvVec{X}^*\rvVec{X}^\T)]
+ O(1) \\
&= N\E[\log(1 + \|\rvVec{X}\|^2)] + O(1).
\end{align}
To bound $h(\rvMat{Y} | \rv{V})$, we use the duality approach~\cite{Moser} as follows
\begin{align}
h(\rvMat{Y} | \rv{V}) &= \E[-\log p(\rvMat{Y} | \rv{V})] \notag\\
&= \E[-\log q(\rvMat{Y} | \rv{V})] - \E_{\rv{V}}[D(\Pc_{\rvMat{Y} | \rv{V} = v} \| \Qc)] \notag \\
&\le \E[-\log q(\rvMat{Y} | \rv{V})],
\end{align}
due to the non-negativity of the Kullback-Leibler divergence
$D(\Pc_{\rvMat{Y} | \rv{V} = v} \| \Qc)$. Here, conditioned on $\rv{V}$,
the distribution $\Pc_{\rvMat{Y} | \rv{V}}$ with probability density function (pdf) $p(.)$ is imposed by the input,
channel, and noise distributions, while $\Qc$ is
any distribution in $\CC^{N\times T}$ with the pdf $q(.)$. 
Note that a proper choice of $\Qc$ 
is the key to a tight upper bound. Our choice is inspired by a
training-based scheme. Specifically, if we send a pilot symbol at time
slot~$v\in\{1,\ldots,T\}$, then the output vector being the sum of
$\rvVec{H}$ and $\rvMat{Z}_{[v]}$ should have comparable power in each
direction since $\rvVec{H}$ is generic by assumption. Therefore, it is
reasonable~(in the DoF sense) to let $\rvMat{Y}_{[v]} \sim
\Rc(N,\Id_N)$, where the family of distributions $\Rc(N,\pmb{A})$ is
defined in Lemma~\ref{lemma:distributionR}. Now, 
$\rvMat{Y}_{[v]}$ should provide a rough estimate of the direction of
the channel vector $\rvVec{H}$. Based on such an observation, it is also
reasonable to assume that, given $\rvMat{Y}_{[v]}$, all other
$\rvMat{Y}_{[i]}$, $i\ne v$, are mutually independent and follow
\begin{align}
\rvMat{Y}_{[i]} &\sim \Rc\left(N,\left(\Id_N+ \rvMat{Y}_{[v]}\rvMat{Y}_{[v]}^\H\right)^{-\frac{1}{2}}\right), \quad \forall i \ne v.
\end{align}
We thus obtain a ``guess'' of the auxiliary joint distribution
$\Qc_{\rvMat{Y}| \rv{V} = v}$.  
\begin{proposition} \label{prop:meanLogQ_p2p}
	With the above choice of auxiliary output distribution, it follows that
	\begin{multline}
	\E[-\log q(\rvMat{Y} | \rv{V})] \le (N+T-1)\E[\log(1+|\rv{X}_\rv{V}|^2)] \\+ N\E[\sum_{i=1,i\ne \rv{V}}^{T}\log\left(1+\frac{|\rv{X}_i|^2}{1+|\rv{X}_\rv{V}|^2}\right)] + O(\log\log P). \label{eq:propp2p}
	\end{multline}
\end{proposition}
\begin{proof}
	See Appendix~\ref{proof:propMeanLogQ_p2p}.
\end{proof}
Plugging the bounds into~\eqref{eq:boundRp2p}, we obtain
\begin{align}
I(\rvVec{X};\rvMat{Y}) &\le (T\!-\!1)\E[\log(1\!+\!|\rv{X}_\rv{V}|^2)] + N\E[\log\frac{1\!+\!|\rv{X}_\rv{V}|^2}{1\!+\!\|\rvVec{X}\|^2}] \notag\\&\hspace{.5cm}+ N\E[\sum_{i=1,i\ne \rv{V}}^{T}\log\left(1+\frac{|\rv{X}_i|^2}{1+|\rv{X}_\rv{V}|^2}\right)] \notag\\&\hspace{.5cm}+ O(\log\log P) \notag\\
&\le (T-1)\log(1+\E[|\rv{X}_\rv{V}|^2]) + O(\log\log P) \\
&\le (T-1)\log^+(P) + O(\log\log P),
\end{align}
where we used the fact that $|\rv{X}_i|^2 \le |\rv{X}_{\rv{V}}|^2 \le \|\rvVec{X}\|^2$, $\forall i \ne \rv{V}$. Thus, the DoF is upper bounded by $\frac{T-1}{T}$, which is tight.

\section{Two-User SIMO MAC} \label{sec:converse}
Let us get back to the MAC in this section and show that, when $T\ge 3, N>1$, any achievable DoF pair $(d_1,d_2)$ must satisfy~\eqref{eq:DoFbound1} and~\eqref{eq:DoFbound2}. 

\subsection{The $T\ge N+1 > 2$ case}
Let us consider the more straightforward case with $T\ge N+1 > 2$. We first bound $R_1$ and $R_2$ using similar techniques as for the single-user case, and then give the tight outer bound for the DoF region in the following steps. 

\subsubsection*{Step 1: Output Rotation and Genie-Aided Bound}
Given $\rvVec{X}_2$, the channel with respect to (w.r.t.) input
$\rvVec{X}_1$ has equivalent noise $\rvVec{H}_2 \rvVec{X}_2^\T +
\rvMat{Z}$. Consider the following eigen-value decomposition
\begin{align}
\rvVec{X}_2^*\rvVec{X}_2^\T = \rvMat{U} \ \diag(0,\dots,0,\|\rvVec{X}_2\|^2) \ \rvMat{U}^\H, 
\end{align}
for some $T\times T$ unitary matrix $\rvMat{U}$. We consider the rotated
output $\Yt = \rvMat{Y} \rvMat{U} = \rvVec{H}_1 \Xvt_1^\T + \Zt$, where
$\Xvt_1^\T = \rvVec{X}_1^\T \rvMat{U} = [\Xt_{11} \ \Xt_{12} \ \dots \
\Xt_{1T}]$ and $\Zt = (\rvVec{H}_2 \rvVec{X}_2^\T +
\rvMat{Z})\rvMat{U}$. Note that given $\rvVec{X}_2$, the first
$T-1$ columns of the noise $\Zt$ are i.i.d.~Gaussian whereas the last
column is stronger as the sum of $\rvVec{H}_2 \|\rvVec{X}_2\|$ and a Gaussian noise
vector. Thus, we have
\begin{align}
T R_1 \le I(\rvVec{X}_1; \rvMat{Y} | \rvVec{X}_2) &= I(\Xvt_1;\Yt | \rvVec{X}_2). \label{eq:boundR1_1}
\end{align}
Let us define the random variable $\rv{V}$ as the index of the strongest
among the first $T-1$ elements of $\Xvt_1$, namely,
\begin{align}
\rv{V} = \arg\max_{i=1,2,\dots,T-1} |\Xt_{1i}|^2.
\end{align}

Similarly as in \eqref{eq:boundRp2p} with the genie-aided bound, 
\begin{align}
\hspace{-.2cm}{I(\Xvt_1;\Yt | \rvVec{X}_2) 
\le h(\Yt | \rvVec{X}_2,\! \rv{V})\! -\! h(\Yt | \Xvt_1,\rvVec{X}_2)\!
+\! \log(T-1).} \label{eq:boundR1_2}
\end{align}

\subsubsection*{Step 2: Bounding $h(\Yt | \Xvt_1,\rvVec{X}_2)$ and $h(\Yt | \rvVec{X}_2, \rv{V})$}
Given $\Xvt_1$ and $\rvVec{X}_2$, we can apply Lemma~\ref{lemma:1} with
$\rvMat{W}=[\rvVec{H}_1\ \rvVec{H}_2\ \rvMat{Z}]$ and $\Am=\left[ 
\rvVec{X}_1 \ \rvVec{X}_2 \ \Id_T \right]^\T\rvMat{U}$ to
obtain
\begin{multline}
h(\Yt|\Xvt_1,\rvVec{X}_2) = N \E[\log\det(
\Am^\H \Am)] + O(1) \\
= N\E\bigg[\log\bigg((1+\|\rvVec{X}_2\|^2)\Big(1+\sum_{i=1}^{T-1}|\Xt_{1i}|^2\Big) + |\Xt_{1T}|^2\bigg)\bigg] \\+ O(1),
\label{eq:bound_ent}
\end{multline}
where the last equality is obtained by applying $\tilde{\rvVec{X}}_1^\T
= \rvVec{X}_1^\T \rvMat{U}$. 

For $h(\Yt | \rvVec{X}_2, \rv{V})$, we use the duality upper bound
as before
\begin{align}
h(\Yt | \rvVec{X}_2, \rv{V}) = \E[-\log p(\Yt | \rvVec{X}_2, \rv{V})]
&\le \E[-\log q(\Yt | \rvVec{X}_2, \rv{V})], \notag
\end{align}
where the only difference from the single-user case is the presence of
$\rvVec{X}_2$. We choose the auxiliary pdf $q(.)$ as follows. Given $\rv{V} = v$, $v\le T-1$, we let $
\Yt_{[v]} \sim \Rc(N,\Id_N)$, 
and given
$\Yt_{[v]}$, the other $\Yt_{[i]}$'s are independent and follow
\begin{align}
\Yt_{[i]} &\sim \Rc\left(N,\left(\Id_N+
\Yt_{[v]}\Yt_{[v]}^\H\right)^{-\frac{1}{2}}\right), \quad
i \not\in\{v, T\}, \label{eq:auxDist_i}\\
\Yt_{[T]} &\sim \Rc\left(N,\left((1+\|\rvVec{X}_2\|^2)\Id_N+ \Yt_{[v]}\Yt_{[v]}^\H \right)^{-\frac{1}{2}}\right).\label{eq:auxDist_T}
\end{align}

\begin{proposition} \label{prop:meanLogQ_MAC}
	With the above choice of auxiliary output distribution, we obtain 
	the upper bound~\eqref{eq:bound_meanlog} for $\E[-\log q(\Yt | \rvVec{X}_2, \rv{V})]$, and hence for $h(\Yt | \rvVec{X}_2, \rv{V})$.
	\begin{figure*}[!t]
		\begin{align}
		\E&[-\log q(\Yt | \rvVec{X}_2, \rv{V})] \le (N+T-2)\E[\log(1+|\Xt_{1\rv{V}}|^2)] + N\E[\sum_{i=1,i\ne V}^{T-1} \log\left(1+\frac{|\Xt_{1i}|^2}{1+|\Xt_{1\rv{V}}|^2}\right)] \notag\\ &\hspace{.2cm}+ N\E[\log(1+\|\rvVec{X}_2\|^2)] + \E[\log\left(1+\frac{|\Xt_{1\rv{V}}|^2}{1+\|\rvVec{X}_2\|^2}\right)] + N\E[\log\left(1+\frac{|\Xt_{1T}|^2}{1+\|\rvVec{X}_2\|^2+|\Xt_{1\rv{V}}|^2}\right)] + O(\log\log P). \label{eq:bound_meanlog}
		\end{align}
		\vspace*{-.8cm}
	\end{figure*}
\end{proposition}
\begin{proof}
	See Appendix~\ref{proof:propMeanLogQ_MAC}.
\end{proof}

\subsubsection*{Step 3: Upper Bounds on $R_1$ and $R_2$}
From~\eqref{eq:boundR1_1}, \eqref{eq:boundR1_2}, \eqref{eq:bound_ent} and \eqref{eq:bound_meanlog}, we have the bound for $R_1$	
\begin{align}
TR_1  
&\le \E[f(\Xvt_1,\rvVec{X}_2)] + O(\log\log P),
\end{align}
where  $f(\Xvt_1,\rvVec{X}_2)$ is defined in~\eqref{eq:f}.
\begin{figure*}[!t]
	\begin{multline}
	f(\Xvt_1,\rvVec{X}_2) \defeq (N+T-2)\log\left(1+\max_{i=1,\dots,T-1}|\Xt_{1i}|^2\right)
	+ \log\left(1+\frac{\displaystyle\max_{i=1,\dots,T-1}|\Xt_{1i}|^2}{1+\|\rvVec{X}_2\|^2}\right) \\
	+ N\log\left(1+\frac{|\Xt_{1T}|^2}{1+\|\rvVec{X}_2\|^2+\displaystyle\max_{i=1,\dots,T-1}|\Xt_{1i}|^2}\right)
	- N\log\left(1+\sum_{i=1}^{T-1}|\Xt_{1i}|^2 + \frac{|\Xt_{1T}|^2}{1+\|\rvVec{X}_2\|^2}\right).
	\label{eq:f}
	\end{multline}
	\setlength{\arraycolsep}{1pt}
	\hrulefill \setlength{\arraycolsep}{0.0em}
	\vspace*{-.1cm}
\end{figure*}
Following the exact same steps by swapping the users' role,
\begin{align}
T R_2  
&\le \E[f(\Xvt_2,\rvVec{X}_1)] + O(\log\log P),
\end{align}
where $\Xvt_2 \defeq \rvVec{X}_2\rvMat{U}_1$ with $\rvMat{U}_1$ from the decomposition 
\begin{align}
\rvVec{X}_1^*\rvVec{X}_1^\T = \rvMat{U}_1 \ \diag(0,\dots,0,\|\rvVec{X}_1\|^2) \ \rvMat{U}_1^\H.
\end{align}%

It follows that, for any $\lambda_1,\lambda_2\ge0$, we have the following upper bound on
the weighted sum rate
\begin{align}
  &\hspace{-.3cm}{\lambda_1 R_1\! + \! \lambda_2 R_2}\notag\\
  &\le \frac{1}{T} \E[ \lambda_1
  f(\Xvt_1,\rvVec{X}_2)\! +\! \lambda_2 f(\Xvt_2,\rvVec{X}_1) ]\! +\! O(\log\log P) \\
  &\le \frac{1}{T} \sup_{\xv_1,\xv_2}[ \lambda_1 f(\tilde{\xv}_1,\xv_2)
  \! + \! \lambda_2 f(\tilde{\xv}_2,\xv_1) ]\! +\! O(\log\log P),
  \label{eq:weighted-sum-rate}
\end{align}%
where the supremum is over all $\xv_1,\xv_2$ subject to the peak power
constraints $\|\xv_1\|^2\le P$ and $\|\xv_2\|^2\le P$. 

\subsubsection*{Step 4: DoF upper bounds}
Since we are only interested in the pre-log at high SNR, it is without
loss of optimality to let 
$\|\xv_1\|^2 = P^{\eta_1}, \|\xv_2\|^2 = P^{\eta_2}$ for some
$\eta_1, \eta_2 \le 1$. In addition, we assume that  
\begin{align}
  \max_{i=1,\dots,T-1}|\tilde{x}_{1i}|^2 &= P^{\bar{\eta}_1}, \quad |\tilde{x}_{1T}|^2 = P^{\eta_{1T}}, \\
\max_{i=1,\dots,T-1}|\tilde{x}_{2i}|^2 &= P^{\bar{\eta}_2}, \quad
|\tilde{x}_{2T}|^2 = P^{\eta_{2T}}.
\end{align}
Hence, at high SNR, 
$\eta_1 = \max\{\bar{\eta}_1,\eta_{1T}\}$, $\eta_2 =
\max\{\bar{\eta}_2,\eta_{2T}\}$. From \eqref{eq:f} and
\eqref{eq:weighted-sum-rate}, we have the weighted sum DoF bound
\begin{multline}
{\lambda_1 d_1 + \lambda_2 d_2} \\ 
\le \lambda_1\! \frac{N\!+\!T\!-\!2}{T}\bar{\eta}_1 +\! \lambda_1\!
\frac{1}{T}(\bar{\eta}_1\!-\!\eta_2)^+ +\! \lambda_1\!
\frac{N}{T}(\eta_{1T}\!-\!\max\{\bar{\eta}_1,\eta_2\})^+ \\-\lambda_1
\frac{N}{T}\max\{\bar{\eta}_1,\eta_{1T}-\eta_2\}\\ 
+ \lambda_2 \frac{N\!+\!T\!-\!2}{T}\bar{\eta}_2 + \lambda_2 \frac{1}{T}(\bar{\eta}_2\!-\!\eta_1)^+ + \lambda_2 \frac{N}{T}(\eta_{2T}-\!\max\{\bar{\eta}_2,\eta_1\})^+ \\-\lambda_2
\frac{N}{T}\max\{\bar{\eta}_2,\eta_{2T}-\eta_1\}, \label{eq:weighted-sum-dof}
\end{multline}
subject to the constraints $\bar{\eta}_1, \eta_{1T} \le 1$ and
$\bar{\eta}_2, \eta_{2T} \le 1$. Taking $(\lambda_1,\lambda_2)$ as $\left(1,\frac{1}{T-2}\right)$ or $\left(\frac{1}{T-2},1\right)$, we can verify that, when  
$3\le N+1\le T$, \eqref{eq:DoFbound1} and~\eqref{eq:DoFbound2} hold for all $(d_1,d_2)$ satisfying~\eqref{eq:weighted-sum-dof}. Thus the optimal DoF region is characterized.

\subsection{The $3\le T \le N$ case}
When $T \le N$, the above choice of auxiliary output distribution is not sufficient for a tight DoF outer bound. To see this, let us take $(\lambda_1,\lambda_2) = \left(1,\frac{1}{T-2}\right)$, then if $\bar{\eta}_1+\eta_2 \ge \eta_{1T} = 1$ and $\eta_2 = \bar{\eta}_1$, \eqref{eq:weighted-sum-dof} becomes
\begin{align}
d_1+\frac{d_2}{T-2} \le \frac{T-1}{T} \bar{\eta}_1 + \frac{N}{T} (\eta_{1T} -\bar{\eta}_1),
\end{align}
which is loose since the right-hand side is larger than $1-\frac{1}{T}$ if $N \ge T$. 
Generally, the bound~\eqref{eq:weighted-sum-dof} can be loose when ${\eta}_{1T} > \max\{\bar{\eta}_1,\eta_2\}$ or ${\eta}_{2T} > \max\{\bar{\eta}_2,\eta_1\}$. To account for such scenarios, we ought to refine our choice of auxiliary output distribution for the duality upper bound. First, given $\rvVec{X}_2$, we define a pair of random variables $(\rv{V},\rv{U})$ as 
\begin{align}
\rv{V} = \arg\max_{i=1,2,\dots,T} \frac{|\Xt_{1i}|^2}{\sigma_i^2},
\end{align}
where $\sigma_i^2 = 1, \forall i <T$ and $\sigma_T^2 = 1+\|\rvVec{X}_2\|^2$, and
\begin{align}
\rv{U} = \begin{cases}
1, &\text{if $|\Xt_{1T}|^2 \ge \max\left\{\displaystyle\max_{i=1,\dots,T-1}|\Xt_{1i}|^2,1\!+\!\|\rvVec{X}_2\|^2\right\}$}, \\
0, &\text{otherwise}.
\end{cases}
\end{align}
Thus, $\Xt_{1\rv{V}}$ is the input entry with the largest instantaneous SNR, and $\rv{U}$ determines a specific configuration of input entry powers in which the choice of auxiliary output distribution in the previous case possibly fails.
Then similarly as for the case $T\ge N+1$, with output rotation, genie-aided bound, and duality upper bound, we have that
\begin{multline}
TR_1 \le I(\Xvt_1;\Yt | \rvVec{X}_2) \\
	\le \E[\!-\!\log q(\Yt | \rvVec{X}_2, \rv{V}, \rv{U})] \!-\!  h(\Yt | \Xvt_1,\rvVec{X}_2)\!
	+\! \log(2T), \label{eq:boundR1_3}
\end{multline}
where $h(\Yt | \Xvt_1,\rvVec{X}_2)$ was calculated in~\eqref{eq:bound_ent}. For $\E[-\log q(\Yt | \rvVec{X}_2, \rv{V}, \rv{U})]$, we choose the auxiliary pdf $q(\Yt | \rvVec{X}_2, \rv{V}, \rv{U})$ as follows. Given $\rv{V} = v$ and $\rv{U} = u$, if $v = T$ or $\{v <T, u = 0\}$, we let $\Yt_{[v]} \sim \Rc(N,\Id_N)$ and conditioned on
$\Yt_{[v]}$, the other $\Yt_{[i]}$'s are independent and follow 
\begin{align}
\Yt_{[i]} \sim \Rc\left(N,\left(\sigma_i^2\Id_N+
\frac{\Yt_{[v]}\Yt_{[v]}^\H}{\sigma_v^2}\right)^{-\frac{1}{2}}\right), \quad i \ne v. \label{eq:tmp636}
\end{align}
This choice is inspired by a training-based scheme in which the input symbol with strongest SNR is used as pilot. After some manipulations similar as for Propositions~\ref{prop:meanLogQ_p2p} and~\ref{prop:meanLogQ_MAC}, we get the bounds~\eqref{eq:tmp629} and~\eqref{eq:tmp631}. 
\begin{figure*}[!t]
	\begin{align}
	\E[\!-\!\log q(\Yt | \rvVec{X}_2, \rv{V} \!=\! v < T, \rv{U} = 0)] &\!\le\! (N\!+\!T\!-\!2)\E[\log(1\!+\!|\Xt_{1v}|^2)] \!+\! 
	N\E[\log(1\!+\!\|\rvVec{X}_2\|^2)] \!+\! \E[\log\left(1\!+\!\frac{|\Xt_{1v}|^2}{1\!+\!\|\rvVec{X}_2\|^2}\right)] 
	\notag\\ &\hspace{8cm}+ O(\log\log P). \label{eq:tmp629} \\
	\E[\!-\!\log q(\Yt | \rvVec{X}_2, \rv{V} \!=\! T)] &\!\le\! N \E[\log(1+\|\rvVec{X}_2\|^2 + |\Xt_{1T}|^2)] + (T-1)\E[\log\left(1+\frac{|\Xt_{1T}|^2}{1+\|\rvVec{X}_2\|^2}\right)] + O(\log\log P). \label{eq:tmp631} 
	\end{align}
	\vspace*{-.9cm}
\end{figure*}
If $\{v <T, u = 1\}$, we let $\Yt_{[v]} \sim \Rc(N,\Id_N)$ and given $\Yt_{[v]}$, the other $\Yt_{[i]}$'s are independent with
\begin{align}
\Yt_{[i]} &\sim \Rc\left(N,\left(\Id_N+
\Yt_{[v]}\Yt_{[v]}^\H\right)^{-\frac{1}{2}}\right), \quad i \not\in\{v, T\}, \\
\Yt_{[T]} &\sim \Rc\left(N,\left((1+\|\rvVec{X}_2\|^2)\Id_N+ \frac{P}{\|\Yt_{[v]}\|^2}\Yt_{[v]}\Yt_{[v]}^\H \right)^{-\frac{1}{2}}\right), 
\end{align}
where the only difference from~\eqref{eq:tmp636} is the presence of the factor $\frac{P}{\|\Yt_{[v]}\|^2}$. This factor is added to account for the fact that when $u = 1$, $|\Xt_{1v}|^2 < |\Xt_{1T}|^2$, which can make the power of $\Yt_{[v]}$ inferior to that of $\Yt_{[T]}$. In this case, we have the bound~\eqref{eq:tmp630}. 
\begin{figure*}[!t]
	\begin{align}
	\E[\!-\!\log q(\Yt | \rvVec{X}_2, \rv{V} \!=\! v < T, \rv{U} = 1)] &\!\le\! (N\!+\!T\!-\!2) \E[\log(1\!+\!|\Xt_{1v}|^2)] + N \E[\log\frac{1\!+\!\|\rvVec{X}_2\|^2 \!+\! |\Xt_{1T}|^2}{1+\|\rvVec{X}_2\|^2+P}] \!+\! N \E[\log(1\!+\!\|\rvVec{X}_2\|^2)] \notag\\ &\hspace{3cm}+ \E[\log\left(1+\frac{P}{1+\|\rvVec{X}_2\|^2}\right)] + O(\log\log P). \label{eq:tmp630}
	\end{align}
	\setlength{\arraycolsep}{1pt}
	\hrulefill \setlength{\arraycolsep}{0.0em}
\end{figure*}

These bounds and~\eqref{eq:bound_ent} give us the bound for $R_1$
\begin{align}
TR_1 &\le \E[g(\Xvt_1,\rvVec{X}_2)] + O(\log\log P),
\end{align}
where  $g(\Xvt_1,\rvVec{X}_2)$ is defined in~\eqref{eq:g},
\begin{figure*}[!t]
	\begin{align}
	g(\Xvt_1,\rvVec{X}_2) \defeq \begin{cases}
	&(T-2)\log\left(1+\displaystyle\max_{i=1,\dots,T-1}|\Xt_{1i}|^2\right)
	+ \log\left(1+\frac{\displaystyle\max_{i=1,\dots,T-1}|\Xt_{1i}|^2}{1+\|\rvVec{X}_2\|^2}\right), \\ &\hspace{2cm}\text{if  $\frac{|\Xt_{1T}|^2}{1+\|\rvVec{X}_2\|^2} < \displaystyle\max_{i=1,\dots,T-1}|\Xt_{1i}|^2$ and $|\Xt_{1T}|^2 \le \max\left\{\displaystyle\max_{i=1,\dots,T-1}|\Xt_{1i}|^2,1\!+\!\|\rvVec{X}_2\|^2\right\}$}, \\
	&(T\!-\!2) \log(1\!+\!\displaystyle\max_{i=1,\dots,T-1}|\Xt_{1i}|^2) + N \log\left(\frac{1\!+\!\|\rvVec{X}_2\|^2 \!+\! |\Xt_{1T}|^2}{1+\|\rvVec{X}_2\|^2+P}\right) + \log\left(1\!+\!\frac{P}{1\!+\!\|\rvVec{X}_2\|^2}\right), \\
	&\hspace{2cm}\text{if  $\frac{|\Xt_{1T}|^2}{1+\|\rvVec{X}_2\|^2} < \displaystyle\max_{i=1,\dots,T-1}|\Xt_{1i}|^2$ and $|\Xt_{1T}|^2 > \max\left\{\displaystyle\max_{i=1,\dots,T-1}|\Xt_{1i}|^2,1\!+\!\|\rvVec{X}_2\|^2\right\}$}, \\
	&(T-1)\log\left(1 + \frac{|\Xt_{1T}|^2}{1+\|\rvVec{X}_2\|^2}\right), \quad 
	\text{if  $\frac{|\Xt_{1T}|^2}{1+\|\rvVec{X}_2\|^2} > \displaystyle\max_{i=1,\dots,T-1}|\Xt_{1i}|^2$}. 
	\end{cases} \label{eq:g}
	\end{align}
	\setlength{\arraycolsep}{1pt}
	\hrulefill \setlength{\arraycolsep}{0.0em}
	\vspace*{-.2cm}
\end{figure*}
and the similar bound for $R_2$
\begin{align}
T R_2 &\le \E[g(\Xvt_2,\rvVec{X}_1)] + O(\log\log P).
\end{align}
The rest of the proof follows from a similar weighted sum bound for the rates and the DoFs as done in the previous case.
 
\section{Conclusion} \label{sec:conclusion}
In this work, we have proposed a new tight outer bound on the DoF region
of the two-user non-coherent SIMO MAC with block fading. 
The outer bound region coincides
with the inner bound region achieved by a simple training-based scheme.
We expect to extend the results to the general MIMO case.  


\begin{appendix}
\subsection{Proof of mathematical preliminaries} \label{app:proofs}
\subsubsection{Proof of Lemma~\ref{lemma:1}} \label{proof:lemma1}
Consider the eigen-value decomposition 
$
\Am = \Um \Sigmam \Vm,
$
where $\Um \in \CC^{m\times m}$ and $\Vm \in \CC^{t\times t}$ are unitary matrices, and $\Sigmam = \begin{bmatrix}
\Sigmam'\\\mathbf{0} \end{bmatrix}$ with $\Sigmam' \in \CC^{t\times t}$ a diagonal matrix containing the singular values of $\Am$. Let $\rvMat{W}'= \rvMat{W} \Um$, we have
\begin{align}
h(\rvMat{W}\Am) & = h(\rvMat{W} \Um \Sigmam \Vm) \\
&= h(\rvMat{W}'\Sigmam) \\
&= h(\rvMat{W}'_{[1:t]} \Sigmam') \\
&= h(\rvMat{W}'_{[1:t]}) + n\log|\det(\Sigma')|^2 \\ \label{eq:tmp620}
&= h(\rvMat{W}'_{[1:t]}) + n\log\det(\Am^\H \Am),
\end{align}
where the second equality is because rotation does not change differential entropy; \eqref{eq:tmp620} follows from a change of variables. Next, it follows from $h(\rvMat{W}'_{[1:t]}) + h(\rvMat{W}'_{[t+1:m]}) \ge h(\rvMat{W}') = h(\rvMat{W}) > -\infty$ that 
\begin{align}
h(\rvMat{W}'_{[1:t]}) > -\infty - h(\rvMat{W}'_{[t+1:m]}) > -\infty,
\end{align} 
where $h(\rvMat{W}'_{[t+1:m]}) < \infty$ since the average total power of $\rvMat{W}'_{[t+1:m]}$ is bounded by $\E[\|\rvMat{W}'\|^2_{\rm F}] = \E[\|\rvMat{W}\|^2_{\rm F}] < \infty$. We also have that $h(\rvMat{W}'_{[1:t]}) < \infty$. Therefore, $h(\rvMat{W}'_{[1:t]})$ is bounded by some constant that only depends on the statistics of $\rvMat{W}$. This concludes the proof.

\subsubsection{Proof of Lemma~\ref{lemma:2}} \label{proof:lemma2}
Let $p(.)$ be the density of $\rv{X}$. We introduce an auxiliary distribution with density $q(x) = \left(\frac{1}{\alpha}-1\right) (1+x)^{-1/\alpha}$, $x\ge 0$, with $\alpha < 1$. Then it follows that $h(\rv{X}) + \E[\log(q(x))] = -D(p\|q) \le 0$, which yields
\begin{align}
\E[\log(1+\rv{X})] \ge \alpha h(\rv{X}) + \alpha \log\left(\frac{1}{\alpha}-1\right). \label{eq:tmp651}
\end{align}
If $\E[\rv{X}]\le 1$, then~\eqref{eq:lemma2} holds readily with $\const = -\alpha$. 

If $\E[\rv{X}] > 1$, we have
\begin{align}
h(\rv{X}) &= h\left(\E[\rv{X}]\frac{\rv{X}}{\E[\rv{X}]}\right)\\
&= \log(\E[\rv{X}]) +  h\left(\frac{\rv{X}}{\E[\rv{X}]}\right) \\
&\ge \log(1+\E[\rv{X}]) -1 +h\left(\frac{\rv{X}}{\E[\rv{X}]}\right),
\end{align}
then applying~\eqref{eq:tmp651},~\eqref{eq:lemma2} holds with $\const = \alpha \log\left(\frac{1}{\alpha}-1\right) - \alpha + \alpha h\left(\frac{\rv{X}}{\E[\rv{X}]}\right) > -\infty$.

\subsubsection{Proof of Lemma~\ref{lemma:3}} \label{proof:lemma3}
We prove the lemma by construction. Consider a rate pair $(R_1,R_2)$ achievable with some input pdf $p_{\|\rvVec{X}_1\|^2}(.)$ and $p_{\|\rvVec{X}_2\|^2}(.)$ satisfying the average power constraints $P$. Let us define a new input distribution with the truncated pdf as
\begin{align}
p_{\underline{\rvVec{X}}_i}(x) = \begin{cases}
\frac{p_{\|\rvVec{X}_i\|^2}(x)}{\Pr(\|\rvVec{X}_i\|^2 < P^\beta)}, &~\text{if $x<P^\beta$}, \\
0, &~\text{if $x\ge P^\beta$},
\end{cases}
\end{align}
for $i = 1,2$, with $\beta>1$. For convenience, let us denote the inputs following $p_{\underline{\rvVec{X}}_1}(x)$ and $p_{\underline{\rvVec{X}}_2}(x)$ as $\underline{\rvVec{X}}_1$ and $\underline{\rvVec{X}}_2$, respectively. Then $\underline{\rvVec{X}}_1$ and $\underline{\rvVec{X}}_2$ satisfy the peak power constraint $P^\beta$. Similarly, we define $\bar{\rvVec{X}}_1$ and $\bar{\rvVec{X}}_2$ with pdf
\begin{align}
p_{\bar{\rvVec{X}}_i}(x) = \begin{cases}
\frac{p_{\|\rvVec{X}_i\|^2}(x)}{\Pr(\|\rvVec{X}_i\|^2 \ge P^\beta)}, &~\text{if $x\ge P^\beta$}, \\
0, &~\text{if $x < P^\beta$}.
\end{cases}
\end{align}
Clearly, $\rvVec{X}_i$ equals $\underline{\rvVec{X}}_i$ if $\|\rvVec{X}_i\|^2 < P^\beta$ and $\bar{\rvVec{X}}_i$ otherwise. We define the random variable $\rv{V}$ as
\begin{align}
\rv{V} = \begin{cases}
0, &~\text{if $\|\rvVec{X}_1\|^2 = \|\underline{\rvVec{X}}_1\|^2$ and $\|\rvVec{X}_2\|^2 = \|\underline{\rvVec{X}}_2\|^2$}, \\
1, &~\text{otherwise}.
\end{cases}
\end{align}
By Markov's inequality, 
\begin{multline}
\Pr(\|\rvVec{X}_i\|^2 = \|\bar{\rvVec{X}}_i\|^2) = \Pr(\|\rvVec{X}_i\|^2 \!\ge\! P^\beta) \\ \le \frac{\E[\|\rvVec{X}_i\|^2]}{P^\beta} \le T P^{1-\beta}, \quad i=1,2, \label{eq:markov1}
\end{multline}
then 
\begin{multline}
\Pr(\rv{V}\!=\!1) = 1-\Pr(\|\rvVec{X}_1\|^2 \!=\! \|\underline{\rvVec{X}}_1\|^2) \Pr(\|\rvVec{X}_2\|^2 \!=\! \|\underline{\rvVec{X}}_2\|^2) \\
\le 1-\left(1-TP^{1-\beta}\right)^2 \le 2TP^{1-\beta}.\label{eq:markov2}
\end{multline}
Let the genie give $\rv{V}$ to the receiver, we have that
\begin{align}
TR_1 &\le I(\rvVec{X}_1;\rvMat{Y}|\rvVec{X}_2) \\
&\le I(\rvVec{X}_1;\rvMat{Y},\rv{V}|\rvVec{X}_2) \\
&= I(\rvVec{X}_1;\rvMat{Y}|\rvVec{X}_2,\rv{V}) + I(\rvVec{X}_1;\rv{V}|\rvVec{X}_2) \\
&\le \Pr(\rv{V}=0)I(\rvVec{X}_1;\rvMat{Y}|\rvVec{X}_2,\rv{V}=0) \notag \\&\hspace{1cm}+ \Pr(\rv{V}=1)I(\rvVec{X}_1;\rvMat{Y}|\rvVec{X}_2,\rv{V}=1) + 1, \label{eq:tmp642} \\
&\le I(\underline{\rvVec{X}}_1;\rvMat{Y}|\underline{\rvVec{X}}_2) + \Pr(\rv{V}=1)I(\rvVec{X}_1;\rvMat{Y}|\rvVec{X}_2,\rv{V}=1) + 1, \label{eq:tmp647}
\end{align}
where~\eqref{eq:tmp642} is due to $I(\rvVec{X}_1;\rv{V}|\rvVec{X}_2) \le H(\rv{V}) \le 1$ bits. 
Next, since removing noise and giving CSI increase the rate,
\begin{align}
&\hspace{-.5cm}\Pr(\rv{V}=1)I(\rvVec{X}_1;\rvMat{Y}|\rvVec{X}_2,\rv{V}=1) \notag\\&\le \Pr(\rv{V}=1)I(\rvVec{X}_1;\rvVec{H}_1\rvVec{X}_1^\T + \rvMat{Z} | \rvVec{H}_1,\rv{V}=1) \\
&\le N \Pr(\rv{V}=1) \log(1+\E[\|\rvVec{X}_1\|^2 | \rv{V}=1]) \\
&\le  N \Pr(\rv{V}=1) \log\left(1+\frac{P}{\Pr(\|\rvVec{X}_1\|^2 \ge P^\beta)} \right) \label{eq:tmp807}\\
&\le N \Pr(\rv{V}=1) \log\left(1+P\right) \notag \\&\hspace{1cm}- N \Pr(\rv{V}=1) \log\Pr(\|\rvVec{X}_1\|^2 \ge P^\beta)\\
&=O(P^{1-\beta}\log P^\beta).
\end{align}
where~\eqref{eq:tmp807} is because 
\begin{multline}
\E[\|\rvVec{X}_1\|^2 | \rv{V}=1] \le \E[\|\bar{\rvVec{X}}_1\|^2] 
= \frac{\int_{P^\beta}^{\infty} x p_{\|\rvVec{X}_1\|^2}(x) dx}{\Pr(\|\rvVec{X}_1\|^2 \ge P^\beta)} 
\\\le \frac{\int_{0}^{\infty} x p_{\|\rvVec{X}_1\|^2}(x) dx}{\Pr(\|\rvVec{X}_1\|^2 \ge P^\beta)} 
\le \frac{P}{\Pr(\|\rvVec{X}_1\|^2 \ge P^\beta)},
\end{multline}
and the last equality follows from~\eqref{eq:markov1} and \eqref{eq:markov2}. 
Plugging this into~\eqref{eq:tmp647} yields 
\begin{align}
TR_1 \le I(\underline{\rvVec{X}}_1;\rvMat{Y}|\underline{\rvVec{X}}_2)+ O(P^{1-\beta}\log P^\beta).
\end{align}
Following the same steps by swapping the users' role, we get the bound for $R_2$
\begin{align}
TR_2 \le I(\underline{\rvVec{X}}_2;\rvMat{Y}|\underline{\rvVec{X}}_1)+ O(P^{1-\beta}\log P^\beta).
\end{align} 
Using similar techniques, we can also show that 
\begin{align}
T(R_1+R_2) \le I(\underline{\rvVec{X}}_1,\underline{\rvVec{X}}_2;\rvMat{Y})+ O(P^{1-\beta}\log P^\beta).
\end{align}
Therefore, there exists $(R'_1,R'_2)$ satisfying 
\begin{align}
R'_1+R'_2 &\le \frac{1}{T}I(\underline{\rvVec{X}}_1, \underline{\rvVec{X}}_2; \rvMat{Y}), \\
R'_1 &\le \frac{1}{T}I(\underline{\rvVec{X}}_1; \rvMat{Y} | \underline{\rvVec{X}}_2),\\
R'_2 &\le \frac{1}{T}I(\underline{\rvVec{X}}_2; \rvMat{Y} | \underline{\rvVec{X}}_1),
\end{align}
i.e., achievable with the constructed inputs $\underline{\rvVec{X}}_1$ and $\underline{\rvVec{X}}_2$ satisfying the peak power constraint $P^\beta$, such that~\eqref{eq:boundRpeak} holds. This concludes the proof.

\subsubsection{Proof of Lemma~\ref{lemma:distributionR}} \label{proof:lemmaDistR}
In this proof, all expectations are implicitly w.r.t. $\Pc$. A direct calculation from~\eqref{eq:distR} yields
\begin{multline}
\E[-\log(r_\rvVec{Y}(\rvVec{Y}))] = -\log|\det\Am|^2 + (N-\alpha) \E[\log\|\Am\rvVec{Y}\|^2] \\+ \frac{\E[\|\Am\rvVec{Y}\|^2]}{\beta} + \log\Gamma(\alpha) + \log \beta^\alpha + \log \frac{\pi^N}{\Gamma(N)}.
\end{multline}
When $\beta = \E[\|\pmb{A} \rvVec{Y}\|^2]$ and $\alpha = \frac{1}{\log(\beta)} = \frac{1}{\log(\E[\|\pmb{A} \rvVec{Y}\|^2])}$, this becomes
\begin{align}
&\E[-\log(r_\rvVec{Y}(\rvVec{Y}))] \notag\\&= -\log|\det\Am|^2 + N\E[\log\|\Am\rvVec{Y}\|^2] - \frac{\E[\log\|\Am\rvVec{Y}\|^2]}{\log(\E[\|\pmb{A} \rvVec{Y}\|^2])} \notag \\ &\hspace{.5cm}+ \log\Gamma\left(\frac{1}{\log(\E[\|\pmb{A} \rvVec{Y}\|^2])}\right) + \log \frac{e\pi^N}{\Gamma(N)}, \\
&=  -\log|\det\Am|^2 + N\E[\log\|\Am\rvVec{Y}\|^2] \notag \\ &\hspace{3.5cm} + O(\log\log(\E[\|\Am\rvVec{Y}\|^2])),
\end{align}
where the last equality is because $0<\frac{\E[\log\|\Am\rvVec{Y}\|^2]}{\log(\E[\|\pmb{A} \rvVec{Y}\|^2])}<1$ and 
\begin{align}
\log\Gamma\left(\frac{1}{\log(\E[\|\pmb{A} \rvVec{Y}\|^2])}\right) - \log\log(\E[\|\Am\rvVec{Y}\|^2]) \rightarrow 0,
\end{align}
as $\E[\|\Am\rvVec{Y}\|^2] \to \infty$ due to
\begin{align}
\lim\limits_{x\to\infty} \log\Gamma\left(\frac{1}{x}\right)-\log x &= \lim\limits_{x\to\infty} \log\left(\frac{1}{x}\Gamma\left(\frac{1}{x}\right)\right) \notag\\
&= \lim\limits_{x\to\infty} \log\left(\Gamma\left(1+\frac{1}{x}\right)\right) \notag\\
&= \log(\Gamma(1)) \notag\\
&= 0.
\end{align}

\subsection{Proof of Proposition~\ref{prop:meanLogQ_p2p}} \label{proof:propMeanLogQ_p2p}
Using Lemma~\ref{lemma:distributionR}, it follows that
\begin{align}
&\E[-\log q(\rvMat{Y} | \rv{V}=v)] \notag\\&= N\E[\log\|\rvMat{Y}_{[v]}\|^2] + \sum_{i=1,i\ne v}^{T} \E\Bigg[\log\det\left(\Id_N+ \rvMat{Y}_{[v]}\rvMat{Y}_{[v]}^\H\right)\Bigg. \notag\\&\hspace{.3cm}\!+\! \Bigg. N\log\left\|\left(\Id_N\!+\! \rvMat{Y}_{[v]}\rvMat{Y}_{[v]}^\H\right)^{\!-\!\frac{1}{2}} \rvMat{Y}_{[i]}\right\|^2\Bigg] \!+\! O(\log\log P)\\
&= N\E[\log\|\rvMat{Y}_{[v]}\|^2] + \sum_{i=1,i\ne v}^{T} \E\Bigg[\log\left(1 + \|\rvMat{Y}_{[v]}\|^2\right) \Bigg. \notag\\&\hspace{.3cm}\!+\! \left. N\log\left(\|\rvMat{Y}_{[i]}\|^2\!-\! \frac{|\rvMat{Y}_{[i]}^\H\rvMat{Y}_{[v]}|^2}{1\!+\!\|\rvMat{Y}_{[v]}\|^2} \right)\right] + O(\log\log P) \\
&= (N+T-1)\E[\log(1+\|\rvMat{Y}_{[v]}\|^2)] \notag\\&\hspace{.3cm}+ N\sum_{i=1,i\ne v}^{T} \E\Big[\log(\|\rvMat{Y}_{[i]}\|^2\!+\!\|\rvMat{Y}_{[i]}\|^2\|\rvMat{Y}_{[v]}\|^2\!-\!|\rvMat{Y}_{[i]}^\H\rvMat{Y}_{[v]}|^2) \Big. \notag\\&\hspace{2cm}-\Big.\log(1+\|\rvMat{Y}_{[v]}\|^2)\Big] + O(\log\log P),
\end{align}
where in the second equality, we used the identities $\det(\Id+\uv\vv^\H) = 1+\vv^\H\uv$, $\det(c\pmb{A}) = c^n \det(\pmb{A})$ for $\pmb{A}\in \CC^{n\times n}$, and $\|(\pmb{A}+\uv\vv^\H)^{-1/2} \xv \|^2 = \xv^\H (\pmb{A}+\uv\vv^\H)^{-1} \xv = \xv^\H \left( \pmb{A}^{-1}-\frac{\pmb{A}^{-1} \uv \vv^\H \pmb{A}^{-1}}{1+\vv^\H \pmb{A}^{-1} \uv}\right) \xv$. 

By expanding $\rvMat{Y}_{[1]}, \dots, \rvMat{Y}_{[T]}$, we get that, given $\rvVec{X}$,
\begin{align}
\E_{\rvVec{H},\rvMat{Z}}[\|\rvMat{Y}_{[i]}\|^2] = N(1+|\rv{X}_i|^2), \quad \forall i, 
\end{align}
and
\begin{multline}
\E_{\rvVec{H},\rvMat{Z}}[\|\rvMat{Y}_{[i]}\|^2\|\rvMat{Y}_{[v]}\|^2 - |\rvMat{Y}_{[i]}^\H\rvMat{Y}_{[v]}|^2] \\= (N^2-N)(1+|\rv{X}_v|^2+|\rv{X}_i|^2),\quad i\ne v.
\end{multline}
Then, using Jensen's inequality and Lemma~\ref{lemma:2} (by letting $\alpha$
arbitrarily close to $1$), we get that
\begin{align}
&\E[-\log q(\rvMat{Y} | \rv{V}=v)] \notag\\&\le (N+T-1)\E[\log(1\!+\!N\!+\!N|\rv{X}_v|^2)] \notag\\&\ +\! N\!\sum_{i=1,i\ne \rv{v}}^{T}\E[\log\frac{N\!+\!N|\rv{X}_i|^2\!+\!(N^2\!-\!N)(1\! +\! |\rv{X}_v|^2\!+\!|\rv{X}_i|^2)}{1+N+N|\rv{X}_v|^2}] \notag\\&\ + O(\log\log P) \\
&= (N+T-1)\E[\log(1+|\rv{X}_v|^2)] \notag\\&\ +\! N\sum_{i=1,i\ne v}^{T}\E[\log\left(1\!+\!\frac{|\rv{X}_i|^2}{1\!+\!|\rv{X}_v|^2}\right)] \!+\! O(\log\log P).
\end{align}
Taking expectation over $\rv{V}$, we obtain~\eqref{eq:propp2p}, which concludes the proof.

\subsection{Proof of Proposition~\ref{prop:meanLogQ_MAC}} \label{proof:propMeanLogQ_MAC}
Hence, we obtain from Lemma~\ref{lemma:distributionR},
\begin{align}
&\E[-\log(q(\Yt |\rvVec{X}_2, \rv{V}=v)] \notag\\&= N\E[\log\|\Yt_{[v]}\|^2] + \sum_{i=1,i\ne v}^{T-1}\E[\log\det\left(\Id_N+ \Yt_{[v]}\Yt_{[v]}^\H\right)] \notag\\
&\hspace{.3cm}+ N\sum_{i=1,i\ne v}^{T-1}\E[\log\left\|\left(\Id_N+ \Yt_{[v]}\Yt_{[v]}^\H\right)^{-\frac{1}{2}} \Yt_{[i]}\right\|^2] \notag\\
&\hspace{.3cm}+ \E[\log\det\left((1+\|\rvVec{X}_2\|^2)\Id_N+ \Yt_{[v]}\Yt_{[v]}^\H\right)] \notag\\
&\hspace{.3cm}+ N\E\bigg[\log\Big\|\left((1\!+\!\|\rvVec{X}_2\|^2)\Id_N\!+\! \Yt_{[v]}\Yt_{[v]}^\H\right)^{-\frac{1}{2}} \Yt_{[T]}\Big\|^2\bigg]\! \notag\\
&\hspace{.3cm}+ O(\log\log P) \\
&\le N\E[\log(1\!+\!|\Xt_{1v}|^2)] \!+\! \sum_{i=1,i\ne v}^{T}\!B_i + O(\log\log P), \label{eq:bound_meanlog1}
\end{align}
where 
\begin{align}
B_i \defeq\! \E\left[\log\left(1\!+\! \|\Yt_{[v]}\|^2\right)\!+\! N\log\bigg(\|\Yt_{[i]}\|^2 \!-\! \frac{|\Yt_{[i]}^\H\Yt_{[v]}|^2}{1\!+\!\|\Yt_{[v]}\|^2} \bigg)\right], \notag
\end{align}
for $i \notin \{v,T\}$, and
\begin{multline}
B_T \defeq \E\Bigg[\log\bigg((1+\|\rvVec{X}_2\|^2)^{N}\Big(1 + \frac{\|\Yt_{[v]}\|^2}{1+\|\rvVec{X}_2\|^2}\Big)\bigg) \Bigg.\\\Bigg. +N\log\Bigg(\frac{1}{1+\|\rvVec{X}_2\|^2}\bigg(\|\Yt_{[T]}\|^2- \frac{|\Yt_{[T]}^\H\Yt_{[v]}|^2}{1+\|\rvVec{X}_2\|^2+\|\Yt_{[v]}\|^2} \bigg) \Bigg)\Bigg]. \notag
\end{multline}
By expanding $\Yt_{[1]}, \dots, \Yt_{[T]}$, we get that, given $\rvVec{X}_1$ and $\rvVec{X}_2$,
\begin{multline}
\E_{\rvVec{H}_1,\rvMat{Z}}[\|\Yt_{[v]}\|^2\|\Yt_{[i]}\|^2 - |\Yt_{[i]}^\H\Yt_{[v]}|^2] \\= (N^2-N)\left(1+|\Xt_{1v}|^2 + |\Xt_{1i}|^2\right), \quad i \notin \{v,T\},
\end{multline}
and
\begin{multline}
\E_{\rvVec{H}_1,\rvMat{Z}}[\|\Yt_{[v]}\|^2\|\Yt_{[T]}\|^2 - |\Yt_{[T]}^\H\Yt_{[v]}|^2] \\= (N^2-N)\left((1+\|\rvVec{X}_2\|^2)(1+|\Xt_{1v}|^2) + |\Xt_{1T}|^2\right) \\
\le (N^2-N)(1+\|\rvVec{X}_2\|^2)(1+|\Xt_{1v}|^2 + |\Xt_{1T}|^2).
\end{multline}
Then, applying repeatedly Lemma~\ref{lemma:2} (by letting $\alpha$
arbitrarily close to $1$), and Jensen's inequality,
\begin{align}
B_i &= \E[-(N-1)\log\left(1 + \|\Yt_{[v]}\|^2\right)] \notag\\
&\hspace{.5cm}+ N\E[\log\left(\|\Yt_{[i]}\|^2+\|\Yt_{[i]}\|^2\|\Yt_{[v]}\|^2 -  |\Yt_{[i]}^\H\Yt_{[v]}|^2\right)] \notag\\
&\le \E[-(N-1)\log(1+N+N|\Xt_{1v}|^2)] \notag\\
&\hspace{.5cm}+ N\E[\log(N^2(1\!+\!|\Xt_{1i}|^2)+(N^2\!-\!N)|\Xt_{1v}|^2)]\!+\!O(1), \notag\\
&=\E[\log(1\!+\!|\Xt_{1v}|^2)] \!+\! N\E[\log\Big(1\!+\!\frac{|\Xt_{1i}|^2}{1\!+\!|\Xt_{1v}|^2}\Big)] \!+\! O(1)\label{eq:tmp458}
\end{align}
for $i \notin \{v,T\}$, and
\begin{align}
&B_T = N\E[\log(1+\|\rvVec{X}_2\|^2)] + \E[\log\left(1+\frac{\|\Yt_v\|^2}{1+\|\rvVec{X}_2\|^2}\right)]\notag\\
&\hspace{.2cm}+ N\E[\log\Big(\|\Yt_{[T]}\|^2\! +\! \frac{\|\Yt_{[v]}\|^2\|\Yt_{[T]}\|^2\! -\! |\Yt_{[T]}^\H\Yt_{[v]}|^2}{1+\|\rvVec{X}_2\|^2}\Big)] \notag\\
&\hspace{.2cm}-N\E[\log\left(1+\|\rvVec{X}_2\|^2+\|\Yt_v\|^2\right)] \\ 
&\le N\E[\log(1+\|\rvVec{X}_2\|^2)] + \E[\log\left(1+\frac{N+N|\Xt_{1v}|^2}{1+\|\rvVec{X}_2\|^2}\right)]\notag\\
&\hspace{.2cm}+ N\E\Big[\log\Big(N(1 + \|\rvVec{X}_2\|^2) \Big.\Big. \notag\\&\hspace{2cm}\Big.\Big.+ (N^2-N)(1+|\Xt_{1v}|^2)+N^2|\Xt_{1T}|^2\Big)\Big] \notag\\
&\hspace{.2cm}- N\E[\log(1+\|\rvVec{X}_2\|^2+N+N|\Xt_{1v}|^2)] + O(1) \\
&= N\E[\log(1+\|\rvVec{X}_2\|^2)] + \E[\log\left(1+\frac{|\Xt_{1v}|^2}{1+\|\rvVec{X}_2\|^2}\right)]\notag\\
&\hspace{.2cm}+N\E[\log\left(1+\frac{|\Xt_{1T}|^2}{1+\|\rvVec{X}_2\|^2+|\Xt_{1v}|^2}\right)] + O(1)\label{eq:tmp466}
\end{align}
Plugging \eqref{eq:tmp458} and \eqref{eq:tmp466} into~{\eqref{eq:bound_meanlog1}} then taking expectation over $\rv{V}$, we obtain \eqref{eq:bound_meanlog}, which concludes the proof. 

\end{appendix}

\bibliographystyle{IEEEtran}
\bibliography{IEEEabrv,./biblio}

\end{document}